\documentclass[11pt]{article}
\usepackage[margin=1in]{geometry}
\usepackage{amsfonts,amsthm,amsmath,amssymb}
\usepackage{cite,microtype,graphicx,hyperref}
\usepackage[capitalize,nameinlink]{cleveref}
\usepackage{tikz}
\usetikzlibrary{automata, positioning, arrows}
\usetikzlibrary{shapes.geometric}
\usepackage{algorithm}

\newtheorem{theorem}{Theorem}
\newtheorem{lemma}{Lemma}
\newtheorem{corollary}{Corollary}

\theoremstyle{definition}
\newtheorem{definition}{Definition}
\newtheorem{example}{Example}

\title{Fast Schulze Voting Using Quickselect}

\let\email\texttt
\author{Arushi Arora\thanks{Computer Science Department, University of California, Irvine; \email{arushia2@uci.edu}, \email{eppstein@uci.edu}, \email{randylh@uci.edu}. Research of David Eppstein was supported in part by NSF grant CCF-2212129.}
\and David Eppstein\footnotemark[1]
\and Randy Le Huynh\footnotemark[1]}

\begin{document}

\maketitle

\begin{abstract}
The Schulze voting method aggregates voter preference data using maxmin-weight graph paths, achieving the Condorcet property that a candidate who would win every head-to-head contest will also win the overall election. Once the voter preferences among $m$ candidates have been arranged into an $m\times m$ matrix of pairwise election outcomes, a previous algorithm of Sornat, Vassilevska Williams and Xu (EC '21) determines the Schulze winner in randomized expected time $O(m^2\log^4 m)$. We improve this to randomized expected time $O(m^2\log m)$ using a modified version of quickselect. 
\end{abstract}

\section{Introduction}
In ranked-choice voting, each voter provides a preference ordering of candidates rather than a single preferred candidate, and one of several methods may be used to aggregate these preferences and determine a winner. Prominent among these methods is one introduced by Markus Schulze in 2003 based on maxmin-weight paths in graphs, called the Schulze method or beatpath method. It has the advantage of being a \emph{Condorcet method}: if one candidate would be preferred by a majority of voters in a head-to-head contest against each other candidate, that preferred candidate wins~\cite{Sch-VM-03,Sch-SCW-11,Sch-18}. More generally, if a subset of candidates wins all head-to-head contests against other candidates, the Schulze winner will be in this subset. This property (the \emph{Smith criterion}~\cite{Smi-Econ-73}) distinguishes the Schulze method from instant-runoff voting, the Borda count, and other commonly used ranked-choice methods.

In the Schulze method, voter preference orderings may be incomplete, specifying only a voter's preferences among a subset of candidates, or grouping some candidates together into equally-preferred subsets. From these votes, one can determine, for each pair of candidates, the number of voters who prefer one candidate over the other, and the \emph{margin of victory} that the more-preferred candidate would hold over the less-preferred candidate in an election between only those two candidates. Schulze defines a weighted complete directed graph, in which the vertices are candidates, and an edge is directed from each pairwise winner to each pairwise loser, weighted by these margins of victory. In this graph, Schulze considers all pairwise maxmin-weight paths. The candidates can be partially ordered by the weights of these paths, considering one candidate to beat another when the maxmin weight  from the first to the second is  greater than the maxmin weight in the other direction. When the margins of victory are distinct, this partial order has a unique maximal element; the Schulze method chooses this candidate as the winner of the election, with a more complex tie-breaking procedure used in the rare case of a tie. A recent preprint of Schulze surveys the usage of the Schulze method in governments, political parties, and organizations including the IEEE and ACM~\cite{Sch-18}.

Determining the outcome of an election using the Schulze method has two separate algorithmic steps: aggregating individual ballots to convert voter preference orderings into a graph weighted by the margins of victory, and then determining a winner using maxmin-weight paths in this graph. In recent work, Sornat, Vassilevska Williams and Xu examined the fine-grained complexity of both subproblems, finding algorithms that (to within logarithmic factors) are optimal under standard complexity-theoretic assumptions~\cite{SorVasXu-EC-22}. Our work focuses on the second part, finding a winner from the weighted graph, for which we develop a simple, fast algorithm with fewer of those logarithmic factors.

The most obvious way to determine a Schulze winner would be to compute all-pairs maxmin-weight paths, use comparisons between path weights to order all candidates, and then search for a maximal element of the resulting partial order. The maxmin-weight path in a graph has been studied under many names, including the widest path, maximum capacity path, or bottleneck longest path. Maxmin-weight paths in undirected graphs are widely used for high-bandwidth network routing, and can be found between all pairs of vertices merely by finding a maximum spanning tree and using its paths. In directed graphs, the problem is not so simple, but it has also been well-studied and has multiple applications beyond voting.
Many standard graph shortest path algorithms can be adapted to find maxmin-weight paths, and initial reference implementations of the Schulze method were based on the Floyd--Warshall algorithm for all-pairs shortest paths~\cite{Sch-VM-03,Sch-SCW-11,Sch-18}. This textbook algorithm, published by Floyd in 1962 and closely related to earlier transitive closure algorithms by Roy and Warshall, takes time cubic in the number $m$ of candidates in the election~\cite{Flo-CACM-62,Roy-CR-59,War-JACM-62}.\footnote{We follow Sornat, Vassilevska Williams and Xu~\cite{SorVasXu-EC-22} in using $m$ for the number of candidates and $n$ for the number of voters in an election, despite the conflict between this notation and the usual conventions ($n$ for the number of voters and $m$ for the number of candidates) for graph algorithms in social choice theory and computational social choice communities.}  In 2009, Duan and Pettie~\cite{DuaPet-SODA-09} showed that all-pairs maxmin-weight paths can be computed by an algorithm based on fast matrix multiplication, taking time $O(m^{(3+\omega)/2})$ where $\omega$ is the exponent of fast matrix multiplication. As of 2024, $\omega\le 2.3716$ giving a time bound for all pairs maxmin-weight paths of $O(m^{2.6858})$~\cite{WilXuXu-SODA-24}.

However, as Sornat, Vassilevska Williams and Xu~\cite{SorVasXu-EC-22} observed, although computing all-pairs maxmin-weight paths is sufficient to determine the outcome of the Schulze method, it is not necessary. Instead, they found a faster algorithm for finding a Schulze winner (or winners, in case of ties). Their algorithm takes as input a weighted complete directed graph (representing the votes that would be cast for each candidate in each possible two-candidate election) and produces the Schulze method winner directly, in randomized expected time $O(m^2\log^4 m)$. The algorithm does not use maxmin-weight paths at all, instead using a complex decremental strong connectivity data structure of Bernstein, Probst, and Wulff-Nilsen~\cite{BerProWul-STOC-19}.
Our main result is a simpler, faster algorithm for determining the Schulze winner, taking randomized expected time $O(m^2\log m)$. Our algorithm is an adaptation of another classical algorithm, quickselect~\cite{Hoa-CACM-61,Dev-JCSS-84}, using single-source maxmin-weight paths in a subroutine that replaces the pivoting steps in quickselect. The single logarithmic factor in our time bound comes from the fact that these maxmin-weight path subroutine calls consider the whole given graph, even in later stages of the algorithm in which the field of candidates has been considerably narrowed down. As a subsidiary result, we show that this slowdown appears necessary for this algorithm: there exist inputs for which restricting these subroutine calls to a narrowed field of candidates would produce incorrect results. For the same reason, we cannot apply quicksort in place of quickselect to produce a near-quadratic-time Schulze ordering of all candidates.

Our $O(m^2\log m)$ time bound is essentially optimal up to its single remaining logarithmic factor, as the input to our algorithm has size $O(m^2)$. We should note that under a lower bound (conditional with respect to standard assumptions of fine-grained complexity) proved by Sornat, Vassilevska Williams and Xu~\cite{SorVasXu-EC-22}, the time for this step in determining a Schulze winner will be dominated by the time for converting voter preference orderings into a matrix of pairwise outcomes, which we do not speed up. Moreover, their lower bound holds for determining the Schulze winner from voter preference input, regardless of whether this conversion is made. Nevertheless, our improvement may have some practical applicability to other problems where the matrix of pairwise outcomes is already available, eliminating the need for conversion to a matrix. For example, it may be of interest to perform computational experiments on synthetic matrix data, or to analyze the outcome of a round-robin tournament for a high-scoring sport such as basketball by applying the Schulze method to the matrix of scores obtained by each team in each game.

\section{Preliminaries}
\label{sec:definitions}

\subsection{Top-heavy partial orders and quickselect}

\begin{definition}
A \emph{strict partial order} on a set $S$ is a binary relation $<$ on $S$ that obeys the following properties:
\begin{itemize}
\item Asymmetry: for every $x$ and $y$ in $S$, at most one of $x<y$ and $y<x$ is true. In particular, setting $y=x$, this implies irreflexivity: $x\not< x$.
\item Transitivity: for every $x$, $y$, and $z$  in $S$, if $x\le y$ and $y\le z$, then $x\le z$.
\end{itemize}
We say that a strict partial order is \emph{top-heavy} if it has a unique top element $t$, such that for all $x\in S$, $x\ne t$ implies $x<t$. We say that two elements are \emph{incomparable}, written as $x||y$, when neither $x<y$ nor $x>y$ is true.
\end{definition}

A \emph{linear extension} of a strict partial order is an arrangement of its elements into a sequence $x_0,x_1,\dots x_{n-1}$ such that, whenever $x_i<x_j$ in the partial order, we have $i<j$ in the sequence. In computer science terms, it is a topological ordering of the directed acyclic graph with a directed edge from the smaller element to the larger element for each pair of comparable elements of the partial order.

Our algorithm will use a version of quickselect, specialized to find the top element of a top-heavy strict partial order. It can be defined by the following steps:

\begin{algorithm}[H]
\caption{Quickselect for the top element of a top-heavy strict partial order}
\label{alg:quickselect-top-heavy-partial}
\begin{enumerate}
\item Choose a uniformly random element $p$ of the given set $S$ as a pivot.
\item By comparing each element to $p$, partition $S$ into the three subsets of elements $L=\{x: x<p\}$, $I=\{x:x||p\}$, and $H=\{x:p<x\}.$
\item If $H$ is non-empty, return the result of recursing into $H$.
\item Otherwise, $I$ must equal $\{p\}$, and $p$ must be the unique top element; return $p$.
\end{enumerate}
\end{algorithm}

Consider any fixed linear extension of $S$ (unknown to the algorithm). Then, in each pivoting step, the set $H$ of elements greater than the pivot $p$ must form a subset of the elements that follow $p$ in the linear extension. It follows that, regardless of the structure of the partial order $S$, the distribution of sizes of $H$ (depending on the random choice of pivot) is minorized by the distribution of sizes of the subsets that would be chosen when applying quickselect to a linear order.
Thus, if we could perform constant-time comparisons (not true in our application), the time for this algorithm would be linear, regardless of the partial order. However, our algorithm for Schulze voting will use a modified form of quickselect in which the individual comparisons of the pivoting step are replaced by a maxmin-weight path computation. This replacement will have significant effects on runtime because it will involve maxmin-weight paths in the entire input graph rather than being restricted to a recursive subproblem.

\subsection{The Schulze method}

Here we describe the Schulze method and some of its basic properties. We do not claim any originality for our observations about the method in this section. Although intended for use with preference ballots, and for the pairwise vote differentials in head-to-head contests between all pairs of candidates, the Schulze method does not require all voters to specify their preference between all pairs of candidates.

\begin{definition}
For a given system of candidates $M$, voters $N$, and strict partial order voter preferences, let $P(x,y)$ denote the number of voters who prefer candidate $x$ to candidate $y$. Define the \emph{weighted majority graph} to be a graph $G_{M,N,P}$ with $M$ as vertices, and with a directed edge from $x$ to $y$ for every pair $(x,y)$ of candidates with $x\ne y$. Label the edge from $x$ to $y$ with the weight $P(x,y)-P(y,x)$.
\end{definition}

$G_{M,N,P}$ is a complete directed graph with antisymmetric integer weights, but we will not use these properties, and these are the only properties that it has. Every weighted graph of this form can be realized as a graph $G_{M,N,P}$, for a large enough set $N$ of voters, even when requiring the voters to list a total order for all candidates (a requirement not made in the Schulze method), by adapting a method of McGarvey~\cite{McG-Econ-53} for constructing preference systems that produce arbitrary patterns of pairwise outcomes. To do this, choose a parameter $k$ such that $2k$ is at least the largest pairwise margin of victory, and construct a pool of $N=2k\tbinom{M}{2}$ voters. For each pair of candidates, choose $2k$ of these voters to prefer the two candidates as their top two choices, with the desired pairwise margin of victory. For $k$ of these voters, choose an arbitrary ordering of the remaining candidates, and for the other $k$ voters, use the reverse of the same ordering, so that the preferences for all other candidates cancel out.

\begin{definition}
For a weighted directed graph $G$ with a given set of candidates as its vertices, we define the \emph{beatpath strength} $B(x,y)$ of candidate $x$ against candidate $y$ as the minimum weight of an edge on a path in $G$ from $x$ to $y$, if such a path exists, with the path chosen to maximize this minimum weight. If $x=y$ we define the beatpath strength to be $B(x,x)=+\infty$ and if $x\ne y$ and no path from $x$ to $y$ exists we define the beatpath strength to be $B(x,y)=-\infty$.

Define the \emph{Schulze order} on the vertices of $G$ by the relation $<$, where for vertices $x$ and $y$ of $G$, we define $x < y$ whenever $B(x,y)< B(y,x)$.
\end{definition}

The original version of Schulze's method~\cite{Sch-VM-03} omits the negatively-weighted edges from $G_{M,N,P}$, but this merely complicates the definition of the Schulze order by making some values $B(x,y)$ undefined, without making any difference in the resulting Schulze order, as each two candidates have at least one non-negative path in one direction.

\begin{lemma}[Schulze~\cite{Sch-SCW-11}, Section 4.1]
The Schulze order is a strict partial order.
\end{lemma}

\begin{lemma}[Schulze~\cite{Sch-SCW-11}, Section 4.2.1]
\label{lem:uniqueness}
If the weights of the edges of a graph $G$ are all distinct, the Schulze order is top-heavy.
\end{lemma}

When the Schulze order is top-heavy, the top element of the order is declared to be the Schulze winner. Sornat, Vassilevska Williams and Xu~\cite{SorVasXu-EC-22} describe an algorithm for finding this Schulze winner in expected time $O(m^2\log^4 m)$, based on the deletion of edges from the graph $G_{M,N,P}$ in ascending order by weight and the use of a data structure for decremental strongly connected components in dynamic graphs.

Although it is likely, with sufficiently many voters, that the weights are indeed distinct, equal weights can occur, and in this case there can be multiple maximal elements in the Schulze order. In such cases, Schulze~\cite{Sch-VM-03} proposes breaking the tie by using the preference ordering of a randomly chosen voter. Rather than implementing this tie-breaking method directly, Sornat, Vassilevska Williams and Xu~\cite{SorVasXu-EC-22} modify their decremental strongly connected component algorithm to find the set of all maximal elements.

\section{Algorithm and its complexity analysis}
We leverage quickselect to identify a Schulze winner without the need of an intermediate step of computing all pairs maxmin-weight paths, and without using decremental strongly connected components.

\subsection{Subroutines}
Although we do not use all pairs maxmin-weight paths, we do nevertheless compute some maxmin-weight paths in the graph $G_{M,N,P}$. For this we use a fast version of Dijkstra's algorithm optimized for dense graphs.

\begin{lemma}
\label{lem:dense-dijkstra}
Single-source or single-destination maxmin-weight paths in a dense graph with $m$ vertices can be computed for a single designated source or destination vertex in expected time $O(m^2)$.
\end{lemma}

\begin{proof}
For single-source maxmin-weight paths, we use Dijkstra's algorithm (with the minimum weight of an edge on a path as its priority rather than the sum of weights), without any priority queue data structure. Instead, we merely maintain the maximum priority found so far for a path to each unprocessed vertex, and in each step choose the next vertex $u$ to process by scanning all unprocessed vertices sequentially and choosing the one whose priority is maximum. Then, as is usual for Dijkstra's algorithm, we consider the paths formed by following one more edge $uv$ from the chosen vertex to each unprocessed vertex $v$, we compute the priority of each such path as the minimum of the priority of $u$ and the weight of edge $uv$, and we update the priority of $v$ to the maximum of its old value and the priority of the path. In pseudocode:

\begin{algorithm}[H]
\caption{Dense single-source maxmin-weight paths}
\label{alg:single-source-max-weight-paths}
\begin{enumerate}
\item Set the priority of every vertex to $-\infty$, and the priority of the source vertex to $+\infty$.
\item Flag each vertex as unprocessed, and initialize a list $U$ of unprocessed vertices.
\item While $U$ is non-empty:
\begin{enumerate}
\item Scan $U$ to find the maximum-priority vertex $u$; let its priority be $\beta$ (the bottleneck weight of the path from the source to $u$).
\item For each edge $uv$ with weight $w$, where $v$ is unprocessed, set the priority of $v$ to the maximum of its old priority with $\min(\beta,w)$.
\item Remove $u$ from $U$
\end{enumerate}
\end{enumerate}
\end{algorithm}

There are $m$ vertices to process, finding the next vertex to process takes time $O(m)$, and processing each vertex takes time $O(m)$, so the total time is $O(m^2)$ as claimed. For single-destination maxmin-weight paths, we apply the same algorithm to the graph obtained by reversing all edges in the given graph.
\end{proof}

Our main idea is that, when a vertex $p$ is selected as the pivot in quickselect, we can use two instances of this single-source or single-destination path computation to perform all comparisons in the Schulze method with respect to $p$.

\begin{lemma}
\label{lem:schulze-pivot}
Let $p$ be any vertex of $G_{M,N,P}$, and let $S$ be any set of vertices. Then in time $O(m^2)$ we can determine, for each vertex $v$ in $S$, whether $v<p$, $v||p$, or $p<v$.
\end{lemma}

\begin{proof}
We use the following algorithm:

\begin{algorithm}[H]
\caption{Pivoting on a single vertex $p$}
\label{alg:single-pivot}
\begin{enumerate}
\item Apply \cref{alg:single-source-max-weight-paths} to $G_{M,N,P}$, using a single-source computation from $p$, to compute each beatpath weight $B(p,v)$.
\item Apply \cref{alg:single-source-max-weight-paths} to to $G_{M,N,P}$ again, using a single-destination computation to $p$, to compute each beatpath weight $B(v,p)$.
\item Partition the candidates into the three sets $L=\{v: B(v,p) < B(p,v)\}$, $I=\{v: B(v,p) = B(p,v)\}$, and $H=\{v: B(v,p) > B(p,v)\}$, and return these three sets.
\end{enumerate}
\end{algorithm}

The fact that these three sets are the sets of elements less than, incomparable to, and greater than the pivot $p$ in the Schulze order follows immediately from the definition of the order. The time bound follows from \cref{lem:dense-dijkstra}.
\end{proof}

\subsection{The main algorithm}

Our main algorithm adapts our quickselect-based algorithm for a maximal element in any top-heavy partial order (\cref{alg:quickselect-schulze}) to use this pivoting algorithm, and then checks for uniqueness. It is convenient to use an iterative version of quickselect rather than a recursive version, but this makes little difference to the algorithm and its correctness.

\begin{algorithm}[H]
\caption{Quickselect for the top element of the Schulze order}
\label{alg:quickselect-schulze}
\begin{enumerate}
\item Let $S$ be the set of all vertices in $G_{M,N,P}$.
\item While $|S| > 1$, do the following steps:
\begin{enumerate}
\item Choose a uniformly random element $p$ of the given set $S$ as a pivot.
\item Use \cref{alg:single-pivot} to partition $S$ into the three subsets of elements $L=\{x: x<p\}$, $I=\{x:x||p\}$, and $H=\{x:p<x\}.$
\item If $H$ is non-empty, set $S=H$; otherwise, set $S=\{p\}$.
\end{enumerate}
\item Let $p$ be the unique member of $S$, guaranteed to be a maximal element of the Schulze order.
\item To test whether $p$ is the unique maximal element, use \cref{lem:schulze-pivot} again to partition $S$ into the three subsets of elements $L=\{x: x<p\}$, $I=\{x:x||p\}$, and $H=\{x:p<x\}.$ $H$ is guaranteed to be empty; the order is top-heavy with $p$ as the unique winner if and only if $I$ is a singleton set.
\end{enumerate}
\end{algorithm}

In the case that the order is not top-heavy, the maximal elements of the Schulze order are exactly the maximal elements of its restriction to the final set $I$ from step 4 of the algorithm. Our algorithm does not determine which elements of $I$ are maximal. There must be at least two of these maximal elements, $p$ and any maximal element of $I\setminus\{p\}$.  It is necessary to recompute $I$ in this step rather than re-using the last such set computed in step 2(b), because this final set may include elements discarded in earlier iterations of the algorithm.

\begin{lemma}
\cref{alg:quickselect-schulze} finds a maximal element of the Schulze order and correctly determines whether it is the unique maximal element.
\end{lemma}

\begin{proof}
By induction on the number of steps of the while loop, in each step the sequence of already-chosen pivots are all comparable in the Schulze order, and $S$ consists of the elements that are greater than all of these pivots. Thus, in the final iteration of the loop, when $H$ becomes empty, there are no elements greater than the final chosen pivot $p$, so $p$ is maximal.

If $p$ is the unique maximal element, then there can be no other element greater than it or incomparable to it,  the sets $H$ and $I$ found in step 4 will be empty and the singleton $\{p\}$ respectively, and the algorithm will correctly report that $p$ is unique. In the other direction, if $p$ is maximal but not the unique maximal element, then any other maximal element will be incomparable with $p$, it will be included in the set $I$ found in step 4, and the algorithm will correctly report that $p$ is not the unique maximal element.
\end{proof}

In \cref{alg:quickselect-schulze}, the main contributor to the total time is the call to \cref{alg:single-pivot} in the inner loop, which takes time $O(m^2)$ according to \cref{lem:schulze-pivot}. Thus, to analyze the algorithm we mainly need to determine how many times this loop iterates. In the next two subsections we determine both the expected runtime of \cref{alg:quickselect-schulze} and high-probability bounds on its runtime, showing a time of $O(m^2\log m)$ in both cases.

\subsection{Expected time}

In both this and the next section our analysis of \cref{alg:quickselect-schulze} is more or less the standard analysis of the usual quickselect algorithm, with the exception that our emphasis is on the number of rounds of iteration made by the algorithm rather than its number of comparisons. This change of emphasis corresponds to the fact that each round takes complexity linear in the size of the entire weighted majority graph, rather than (as in standard quickselect) becoming faster as the number of remaining candidates becomes smaller.

Let $T(m)$ denote the expected runtime of \cref{alg:quickselect-schulze} on a worst-case input of size $m$, and let $R(m)$ denote the expected number of iterations of the while loop of the algorithm, again on a worst-case input of size $m$. Each iteration takes time $O(m^2)$ by \cref{lem:schulze-pivot}, and the work performed outside of the while loop is also $O(m^2)$ for the same reason; thus, $T(m)=R(m) * m^2$. In the analysis of this section, we show that $R(m)=O(\log m)$ and therefore that $T(m)=O(m^2\log m)$. This would also follow from our high-probability analysis but for expected time we can obtain more precise bounds on the number of iterations.

\begin{lemma}
\label{lem:iter-size}
Let $S$ have size $s$ before the start of an iteration of the while loop of the algorithm, and let $S'$ be the new value of $S$ after the iteration. Then, for each possible size $i$ of $S'$ (with $1\le i < s$),
\[\Pr[|S'|\le i]\ge\frac{(i+1)}{s}.\]
\end{lemma}

\begin{proof}
Fix any linear extension of the Schulze order, and number the elements of $S$ as $x_1,x_2,\dots x_s$ from larger to smaller in this extension ordering. Then if pivot $x_j$ is chosen, then either $S'=\{x_j\}$ (with size one) or $S'$ contains only (a subset of the) elements $x_k$ with $k<j$ (with size at most $j-1$). Thus, whenever $j\le i+1$ the size of $S'$ will be at most $i$, and this choice of $j$ happens with probability $(i+1)/s$.
\end{proof}

\begin{corollary}
\label{cor:recurrence}
$R(m)$ obeys the recurrence
\[
R(m)= 1 +\frac1m  \sum_{i=1}^m R(i-1)
\]
with base case $R(0)=R(1)=0$.
\end{corollary}

\begin{proof}
The base case follows as while loop will immediately stop iterating when $S$ has fewer than two elements.
Because $R(m)$ is obviously monotonic in $m$, the worst case size distribution for the size after each iteration, obeying \cref{lem:iter-size}, is the distribution in which size $1$ has probability $2/s$ and each other size has probability $1/s$.
The recurrence of the corollary describes exactly this worst-case size distribution after the trivial substitution of size $0$ and $1$ with probability $1/s$ each in place of size $1$ with probability $2/s$.

We have equality for the recurrence rather than inequality, because this worst case distribution can be achieved, for each iteration of the algorithm, when the input Schulze ordering is a linear ordering. For instance this would be true when there is a single voter with that ordering as their preference.
\end{proof}

\begin{theorem}
$R(m)=\ln m + O(1)$ and $T(m)=O(m^2\log m)$.
\end{theorem}

\begin{proof}
Let $\rho(m)$ obey the same recurrence as \cref{cor:recurrence}, with a changed base case $\rho(0)=0$, $\rho(1)=1$.
Then by induction on $m$, $\rho(m)=\sum_{i=1}^m\tfrac1i$, a harmonic number. To see this, use the induction hypothesis to expand and regroup
\begin{align*}
\rho(m) &= 1+\frac1m\sum_{i=1}^{m-1}\sum_{j=1}^i\frac1j\\
&= 1 + \frac1m\left(1+\sum_{j=2}^{m-1} \left(1+(m-j)\frac1j\right)\right)\\
&= 1 + \frac1m\left(1+\sum_{j=2}^{m-1} \frac{m}{j}\right)\\
&= 1 + \frac1m + \sum_{j=2}^{m-1} \frac1j.\\
\end{align*}

As is well known, the $m$th harmonic number is $\ln m+O(1)$.
The change of the base case decreases the overall value of $R(m)$ from the harmonic numbers but the amount of decrease is less than the change by one unit that would occur when decreasing both base cases by one.
\end{proof}

\subsection{High probability analysis}

Define an iteration of the while loop of \cref{alg:quickselect-schulze} to be a \emph{halving iteration} if, after the iteration, $S$ has decreased to half its former size or smaller. Obviously, the algorithm will terminate after at most $\log_2 m$ halving iterations. By \cref{lem:iter-size}, each iteration has probability at least $1/2$ of being a halving iteration, independently of all previous iterations. We may thus model the process of \cref{alg:quickselect-schulze} by a sequence of random coin flips, stopping when $\log_2 m$ heads are reached. The actual algorithm may stop earlier (because the algorithm can terminate with fewer than $\log_2 m$ halving iterations) but any valid high-probability upper bound on the number of coin flips until stopping will be a valid bound on the number of iterations of the algorithm.

\begin{theorem}
For every $c>0$ there exists $\kappa>0$ such that with probability $\ge 1-\frac1{m^c}$ the algorithm terminates after at most $\kappa\log_2 n$ iterations. Thus, with high probability (polynomially close to one) it takes time $O(m^2\log n)$.
\end{theorem}

\begin{proof}
We use a standard Chernoff bound, in the form that for a sum $X$ of random $0$-$1$ variables with expected value $\mu$,
$$\Pr[X \geq (1+\delta)\mu] \leq \bigg(\frac{e^{\delta}}{(1+\delta)^{1+\delta}}\bigg)^{\mu}$$
for any chosen parameter value $\delta>0$.

Here, we let $X$ be the number of \emph{tails} of the random coin-flip process discussed above, in which we aim to get at least $\log_2 m$ heads. If the process consists of $\kappa\log_2 m$ trials, for a parameter $\kappa$ to be determined,
we can calculate the expected number of tails as $\mu=\tfrac12\kappa\log_2 m$. The process fails to get enough coin flips when
\[X>(\kappa-1)\log_2 m= (1+\delta)\mu,\]
where to achieve the equality in this equation we set $\delta = 1-2/\kappa$.

For large values of $\kappa$, we will have $\delta$ close to one, and failure probability upper-bounded by something close to $(e/4)^\mu$. For instance, whenever $\kappa\ge 10$ we will have
\[\frac{e^{\delta}}{(1+\delta)^{1+\delta}} < 0.7726 \]
and failure probability upper-bounded by
\[ 0.7726^{\tfrac12\kappa\log_2 m} = m^{(\tfrac12\kappa)\log_2 0.7726} < m^{-0.186\kappa}. \]
By setting $\kappa=\max\{10,c/0.186\}$ we can make the failure probability be smaller than any polynomial bound $1/m^c$, as desired.
\end{proof}

\subsection{All maximal elements}
\label{sec:all-maximal}

Although multiple maximal elements are unlikely by \cref{lem:uniqueness}, they can occur when multiple edges of the weighted majority graph have equal weight. As an extreme case, all weights can be zero and all candidates tied. \cref{alg:quickselect-schulze} will detect these situations but will not list all maximal elements of the Schulze order; the set of elements incomparable with its chosen maximal element may be a strict superset of the maximal elements. Sornat, Vassilevska Williams and Xu describe how to modify their algorithm to obtain all maximal elements in randomized expected time $O(m^2\log^4 m)$, the same as their time bound for finding a single maximal element~\cite{SorVasXu-EC-22}. While we do not achieve the same $O(m^2\log m)$ time bound as we do for \cref{alg:quickselect-schulze}, we can find all maximal elements more quickly when (as is likely) there are few of them:

\begin{theorem}
\label{thm:all-maximal}
It is possible to find all maximal elements of the Schulze order in expected time $O(km^2\log m)$, where $k$ is the number of maximal elements found.
\end{theorem}

\begin{proof}
Apply the following algorithm.

\begin{algorithm}[H]
\caption{All maximal elements of the Schulze order}
\begin{enumerate}
\item Set $C$ to be the set of all candidates.
\item While $C$ is not empty:
\begin{enumerate}
\item Use \cref{alg:quickselect-schulze}, with step 1 modified to set $S=C$ rather than setting it to the set of all vertices, to find a maximal element $p$ of $C$, and output $p$.
\item Use \cref{lem:schulze-pivot} to partition $C$ into the three subsets of elements $L=\{x: x<p\}$, $I=\{x:x||p\}$, and $H=\{x:p<x\}.$ By the maximality of $p$, $H$ will be empty.
\item Set $C$ to $I\setminus\{p\}$, the subset of elements not less than or equal to $p$.
\end{enumerate}
\end{enumerate}
\end{algorithm}

By induction, the set $C$ remaining after each iteration of the outer loop is the subset of candidates that are not less than or equal to any of the maximal elements found so far. Therefore, any maximal element in $C$ is a maximal element of the whole Schulze order, and all maximal elements will have been found when $C$ becomes empty. Each iteration of the while loop produces a single maximal element and takes time $O(m^2\log m)$ by the analysis of \cref{alg:quickselect-schulze}, which remains valid for the modified initial choice of $S$.
\end{proof}

\section{Barriers to improvement}
\label{sec:barriers}

Several natural directions for extension, generalization, or speedup of our algorithm are blocked by the counterexamples that we present in this section.

\subsection{Tiebreak order}

Ties can occur in the Schulze method only when two head-to-head contests would have equal margins of victory (\cref{lem:uniqueness}), unlikely when the number of voters is large enough to cause large statistical fluctuations in these numbers. Nevertheless, in practical applications of the Schulze method, ties must be disambiguated by some tie-breaking method. Schulze~\cite{Sch-VM-03} suggests using a randomly-drawn ballot (or sequence of ballots) to determine an ordering among the maximal elements of the Schulze order, without changing the order relations already determined by this order. This can be performed by using \cref{thm:all-maximal} to list all maximal candidates, and then applying Schulze's random balloting procedure to this list. However, compared to using \cref{alg:quickselect-schulze}, this would incur a time penalty proportional to the number of candidates listed. It is natural to hope that the ordering from a random ballot could somehow be incorporated into the faster algorithm of \cref{alg:quickselect-schulze}, for instance by basing it on a partial order that refines the Schulze ordering using the random ballot order. However, as the example below shows, any such refinement cannot be based purely on local information (the Schulze order comparison between a pair of values and their relative position on the random ballot).

\begin{example}
Consider an election with three candidates $a$, $b$, and $c$, and with four voters with preferences $a>b>c$, $c>a>b$, $a>c>b$, and $b>c>a$. Then in the weighted majority graph, shown in \cref{fig:example-1-weighted-majority}, all edges into and out of $c$ have weight zero, while the edge from $a$ to $b$ has weight 2 and its reverse has weight $-2$. The beatpath strengths are the same except that the beatpath strength from $b$ to $a$ is zero (by a path through $c$), shown in \cref{tab:example-1-beatpaths}. Thus, $a$ and $c$ are the two maximal elements. The only comparable pair in the Schulze order is $a>b$; the other two pairs are incomparable.

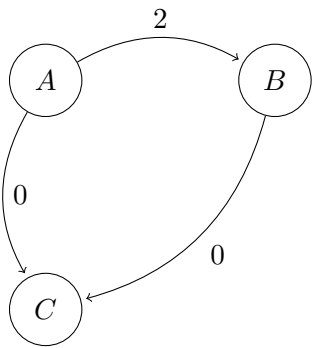
\begin{figure}[H]
\centering
\begin{tikzpicture} [shorten >= 2pt,node distance=1.2in,auto]
    \node[state] (A)  {$A$};
    \node[state] (B) [right of=A] {$B$};
    \node[state] (C) [below of=A] {$C$};
    \path[->]   
                (A) edge [bend left] node {$2$} (B)
                (A) edge [bend right] node {$0$} (C)
                (B) edge [bend left] node {$0$} (C);
\end{tikzpicture}
\caption{Weighted Majority Graph for Example 1}
\label{fig:example-1-weighted-majority}
\end{figure}

\begin{table}[H]
    \centering
    \begin{tabular}{c||ccc}
         & A & B & C \\
         \hline\hline
        A & --- & 2 & 0 \\
        B & 0 & --- & 0 \\
        C & 0 & 0 & --- \\
    \end{tabular}
    \caption{Beatpaths for Example 1}
\label{tab:example-1-beatpaths}
\end{table}

Now, suppose that the randomly chosen ballot is the one with the order $b>c>a$. We cannot determine a partial order by using this order as a tiebreaker between any two incomparable elements of the Schulze order, because this would produce the comparisons $a>b$ (not a broken tie), $b>c$ (by the tiebreak rule), and $c>a$ (by the tiebreak rule), giving a cyclic sequence of comparisons that is not allowed in a strict partial order.
\end{example}

Because Schulze's random ballot tiebreaking method is not consistent with the rest of the method (it does not produce a partial order that can be interpreted as a Schulze order for a perturbed ballot count), it may make sense to instead break ties by perturbing the weighted majority graph: randomly order the edges of this graph and, for an edge in position $i$ of the order, add $i/m^2$ to its weight. These perturbations are small enough that they cannot change the ordering among comparable pairs of candidates in the unperturbed Schulze ordering; they can only make an incomparable pair become comparable. All candidates are treated equally by this perturbation method. It causes all edges to have distinct weights, from which by \cref{lem:uniqueness} there can be only one Schulze winner, the unique maximal element of the Schulze order. This winner can be found in $O(m^2\log m)$ expected time by \cref{alg:quickselect-schulze}.

\subsection{Paths in induced subgraphs}

\cref{alg:quickselect-schulze} takes time $O(m^2\log m)$, rather than $O(m^2)$ (linear in the size of its input) because in each of the $O(\log n)$ iterations of its outer loop it applies a linear-time subroutine on the entire weighted majority graph. In contrast, the standard quickselect algorithm takes linear time, with the same number of iterations (or recursive calls), because later iterations with fewer elements take less time. It is natural to hope that we could obtain a similar speedup in \cref{alg:quickselect-schulze} by considering only maxmin-weight paths in a smaller subgraph, such as the induced subgraph of the previous pivot and the remaining candidates. As the next example shows, however, restricting to the induced subgraph in this way would produce incorrect results.

\begin{example}
Consider a weighted majority graph with distinct positive edge weights as shown in \ref{fig:example-2-weighted-majority}. In a weighted majority graph, the reversed edges would have negative weights, but as the positive edges form a strongly connected subgraph these negative weights will not be used in any maxmin-weight path. As discussed earlier, every weighted complete directed graph with antisymmetric integer weights, including this graph, is the weighted majority graph for some system of voters and candidates, but we do not explicitly construct ballots that would produce this graph.
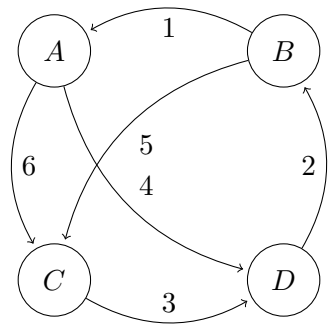
\begin{figure}[H]
\centering
\begin{tikzpicture} [shorten >= 2pt,node distance=1.2in,auto]
    \node[state] (A)  {$A$};
    \node[state] (B) [right of=A] {$B$};
    \node[state] (C) [below of=A] {$C$};
    \node[state] (D) [below of=B] {$D$};
    \path[->]   
                (B) edge [bend right] node {$1$} (A)
                (D) edge [bend right] node {$2$} (B)
                (C) edge [bend right] node {$3$} (D)
                (A) edge [bend right] node {$4$} (D)
                (A) edge [bend right] node {$6$} (C)
                (B) edge [bend right] node {$5$} (C);
\end{tikzpicture}
\caption{Weighted Majority Graph for Example 2}
\label{fig:example-2-weighted-majority}
\end{figure}

Calculating maxmin-weight paths between each two candidates produces the table of beatpath strengths shown below in \cref{tab:example-2-beatpaths}.

\begin{table}[H]
    \centering
    \begin{tabular}{c||cccc}
         & A & B & C & D \\
         \hline\hline
        A & --- & 2 & 6 & 4\\
        B & 1 & --- & 5 & 3\\
        C & 1 & 2 & --- & 3\\
        D & 1 & 2 & 2 & ---\\
    \end{tabular}
        \caption{Beatpaths for Example 2}
\label{tab:example-2-beatpaths}
\end{table}

These strengths produce the Schulze ordering $A > B > C > D$ (a total order), in which $A$ is the unique maximal element and the unique Schulze winner.

Now suppose that \cref{alg:quickselect-schulze} selects $C$ as its first random pivot. If we modified the algorithm to compute subsequent paths in the induced subgraph of  $\{A,B,C\}$ or of $\{A,B\}$, we would still have a path from $B$ to $A$ with maxmin weight 1; however, the removal of $D$ has eliminated all positive paths from $A$ to $B$, leaving the edge from $A$ to $B$ as the maxmin-weight path, with weight $-1$. Thus, the Schulze order for either induced subgraph would have $B>A$, different from the Schulze order of the whole graph. The modified algorithm that finds paths in either induced subgraph would produce an incorrect result.
\end{example}

\subsection{Structure in the Schulze order}

Given our use of quickselect to find a maximal element in the Schulze order, it is natural to hope that a similar variation of quickselect might be used to find elements in intermediate positions in the order, or that a variation of quicksort might be used to determine the entire order. These hopes would be boosted if, for instance, the Schulze order were to turn out to be a weak order, a partition of the candidates into tied sets with a total ordering on those tied sets. Weak orders are commonplace in applications of the standard quicksort and quickselect algorithms (for which the input might be a collection of records with numerical priorities that can be tied) and the algorithms are easily adapted to this case.

Some structure in the Schulze order is provided by \cref{lem:uniqueness}, according to which (for graphs with distinct edge weights) there is a unique maximal element. The same lemma, applied to an edge-reversed version of the same graph, shows also that their Schulze orders have a unique minimal element. However, that is the only structure available in these orders, because every partial order with unique maximal and minimal elements can be realized in this way:

\begin{theorem}
Let $M$ be a set of candidates with a strict partial order $<$ for which there is a unique maximal element $t$ and a unique minimal element $b$. Then there exist voters $N$ and preferences $P$ for which the Schulze order of the weighted majority graph $G_{M,N,P}$ is exactly the given strict partial order.
\end{theorem}

\begin{proof}
$G_{M,N,P}$ must be a complete directed graph on the given candidate set $M$; it remains to assign weights to its edges in order to make the given order be the Schulze order. To do so, we use the following construction:

\begin{example}
Form pools of distinct positive weights to assign to the edges, which we call \emph{large}, \emph{medium}, \emph{small}, and \emph{tiny}, with all weights ordered as $\text{large} > \text{medium} > \text{small} > \text{tiny}$, and with enough distinct weights in each pool to assign to each edge of $G_{M,N,P}$. We will assign these weights in such a way that the comparisons between pairs of edges in the same pool as each other will not affect the resulting Schulze order.
For each candidate $x$ or pair of candidates $x$ and $y$ in $M\setminus\{t,b\}$, we assign these weights to the edges of $G_{M,N,P}$ as follows:
\begin{itemize}
\item We assign a large weight to each edge from $t$ to $x$, and to each edge from $x$ to $b$ (and the negation of the same weight to the opposite edge).
\item If $x > y$, we assign a medium weight to the edge from $x$ to $y$, and the negation of the same weight to the opposite edge.
\item We assign a small weight to the edge from $b$ to $t$, and its negation to the opposite edge.
\item If $x || y$, we assign a tiny weight to one of the two edges between $x$ and $y$, and its negation to the opposite edge.
\end{itemize}
\end{example}

With these weights, $t$ has a large-weight path to every other candidate, and the only positive incoming edge to $t$ has small weight, so $t$ is the (unique by \cref{lem:uniqueness}) maximal element of the Schulze order, matching the given partial order. Symmetrically, $b$ has a large-weight path from every other element, and the only positive outgoing edge from $b$ has small weight, so $b$ is the unique minimal element of the Schulze order, matching the given partial order.

Each remaining pair of candidates $x$ and $y$ can be connected by a path $x$--$b$--$t$--$y$, in which the minimum weight edge is the edge from $b$ to $t$ with small weight. There is no large-weight path between them, and a medium-weight path exists if and only if $x>y$. Thus, the maxmin-weight path has medium weight if $x>y$, and otherwise it has small weight equal to the weight of the edge from $b$ to $t$. When $x > y$ in the given order, this gives $x > y$ in the Schulze order. When $x||y$ in the given order, this gives equal weight to the maxmin-weight paths from $x$ to $y$ and from $y$ to $x$, giving $x||y$ in the Schulze order. Thus in all cases the Schulze order matches the given order, as stated.
\end{proof}

\section{Conclusions}
\label{sec:conclusions}
We have presented a novel algorithm for computing a Schulze winner with an improved time complexity of $O(m^2\log m)$, given as input a weighted majority graph, and a variant of the algorithm that finds all maximal elements of the Schulze order (in case of ties) in expected time $O(km^2\log m)$, compared to a previous $O(m^2\log^4 m)$ time bound \cite{SorVasXu-EC-22}.

As our algorithm is randomized (like the previous algorithm) it is natural to ask for the best time bound of a deterministic algorithm. Additionally, although the examples in \cref{sec:barriers} provide barriers to certain natural directions in which our algorithm might be improved, they still leave open the possibility that an entirely different algorithm might achieve a faster time bound.

\bibliographystyle{plainurl}
\bibliography{bib}

\begin{thebibliography}{10}

\bibitem{BerProWul-STOC-19}
Aaron Bernstein, Maximilian Probst, and Christian Wulff-Nilsen.
\newblock {Decremental strongly-connected components and single-source reachability in near-linear time}.
\newblock In Moses Charikar and Edith Cohen, editors, {\em Proceedings of the 51st Annual ACM SIGACT Symposium on Theory of Computing, STOC 2019, Phoenix, AZ, USA, June 23-26, 2019}, pages 365{--}376. ACM, 2019.
\newblock \href {https://doi.org/10.1145/3313276.3316335} {\path{doi:10.1145/3313276.3316335}}.

\bibitem{Dev-JCSS-84}
Luc Devroye.
\newblock {Exponential bounds for the running time of a selection algorithm}.
\newblock {\em J. Comput. System Sci.}, 29(1):1{--}7, 1984.
\newblock \href {https://doi.org/10.1016/0022-0000(84)90009-6} {\path{doi:10.1016/0022-0000(84)90009-6}}.

\bibitem{DuaPet-SODA-09}
Ran Duan and Seth Pettie.
\newblock {Fast algorithms for (max, min)-matrix multiplication and bottleneck shortest paths}.
\newblock In Claire Mathieu, editor, {\em Proceedings of the Twentieth Annual ACM-SIAM Symposium on Discrete Algorithms, SODA 2009, New York, NY, USA, January 4-6, 2009}, pages 384{--}391. SIAM, 2009.
\newblock \href {https://doi.org/10.1137/1.9781611973068.43} {\path{doi:10.1137/1.9781611973068.43}}.

\bibitem{Flo-CACM-62}
Robert~W. Floyd.
\newblock {Algorithm 97: Shortest path}.
\newblock {\em Communications of the ACM}, 5(6):345, 1962.
\newblock \href {https://doi.org/10.1145/367766.368168} {\path{doi:10.1145/367766.368168}}.

\bibitem{Hoa-CACM-61}
C.~A.~R. Hoare.
\newblock {Algorithm 65: Find}.
\newblock {\em Communications of the ACM}, 4(7):321{--}322, July 1961.
\newblock \href {https://doi.org/10.1145/366622.366647} {\path{doi:10.1145/366622.366647}}.

\bibitem{McG-Econ-53}
David~C. McGarvey.
\newblock {A theorem on the construction of voting paradoxes}.
\newblock {\em Econometrica}, 21:608{--}610, 1953.
\newblock \href {https://doi.org/10.2307/1907926} {\path{doi:10.2307/1907926}}.

\bibitem{Roy-CR-59}
Bernard Roy.
\newblock {Transitivit{\'e} et connexit{\'e}}.
\newblock {\em C. R. Acad. Sci. Paris}, 249:216{--}218, 1959.

\bibitem{Sch-18}
Markus Schulze.
\newblock {The Schulze Method of Voting}.
\newblock Electronic preprint arxiv:1804.02973.

\bibitem{Sch-VM-03}
Markus Schulze.
\newblock {A new monotonic and clone-independent single-winner election method}.
\newblock {\em Voting Matters}, 17:9{--}19, 2003.
\newblock URL: \url{https://www.mcdougall.org.uk/VM/ISSUE17/I17P3.PDF}.

\bibitem{Sch-SCW-11}
Markus Schulze.
\newblock {A new monotonic, clone-independent, reversal symmetric, and Condorcet-consistent single-winner election method}.
\newblock {\em Social Choice and Welfare}, 36(2):267{--}303, 2011.
\newblock \href {https://doi.org/10.1007/s00355-010-0475-4} {\path{doi:10.1007/s00355-010-0475-4}}.

\bibitem{Smi-Econ-73}
John~H. Smith.
\newblock {Aggregation of preferences with variable electorate}.
\newblock {\em Econometrica}, 41(6):1027{--}1041, 1973.
\newblock \href {https://doi.org/10.2307/1914033} {\path{doi:10.2307/1914033}}.

\bibitem{SorVasXu-EC-22}
Krzysztof Sornat, Virginia Vassilevska~Williams, and Yinzhan Xu.
\newblock {Fine-grained complexity and algorithms for the Schulze voting method}.
\newblock In {\em Proceedings of the 22nd ACM Conference on Economics and Computation (EC '21), Budapest, Hungary}, pages 841{--}859. Association for Computing Machinery, 2021.
\newblock \href {https://doi.org/10.1145/3465456.3467630} {\path{doi:10.1145/3465456.3467630}}.

\bibitem{War-JACM-62}
Stephen Warshall.
\newblock {A theorem on Boolean matrices}.
\newblock {\em Journal of the ACM}, 9:11{--}12, 1962.
\newblock \href {https://doi.org/10.1145/321105.321107} {\path{doi:10.1145/321105.321107}}.

\bibitem{WilXuXu-SODA-24}
Virginia~Vassilevska Williams, Yinzhan Xu, Zixuan Xu, and Renfei Zhou.
\newblock {New bounds for matrix multiplication: from alpha to omega}.
\newblock In David~P. Woodruff, editor, {\em Proceedings of the 2024 ACM{--}SIAM Symposium on Discrete Algorithms, SODA 2024, Alexandria, VA, USA, January 7{--}10, 2024}, pages 3792{--}3835. SIAM, 2024.
\newblock \href {https://doi.org/10.1137/1.9781611977912.134} {\path{doi:10.1137/1.9781611977912.134}}.

\end{thebibliography}

\end{document}